\documentclass[11pt]{article}

\usepackage[T1]{fontenc}
\usepackage[utf8]{inputenc}
\usepackage{lmodern}
\usepackage[margin=1in]{geometry}
\usepackage{amsmath,amsfonts,amssymb,amsthm,mathtools,bm}
\usepackage{microtype}
\usepackage{booktabs}
\usepackage{enumitem}
\usepackage{xcolor}
\usepackage[round,authoryear]{natbib}
\usepackage{url}
\usepackage[hidelinks]{hyperref}
\usepackage[nameinlink,capitalise,noabbrev]{cleveref}
\allowdisplaybreaks

\newcommand{\Tr}{\operatorname{Tr}}

\newtheorem{theorem}{Theorem}
\newtheorem{proposition}[theorem]{Proposition}
\newtheorem{lemma}[theorem]{Lemma}
\newtheorem{corollary}[theorem]{Corollary}
\newtheorem{definition}[theorem]{Definition}

\crefname{theorem}{Theorem}{Theorems}
\crefname{proposition}{Proposition}{Propositions}
\crefname{lemma}{Lemma}{Lemmas}
\crefname{corollary}{Corollary}{Corollaries}
\crefname{definition}{Definition}{Definitions}
\crefname{assumption}{Assumption}{Assumptions}
\crefname{remark}{Remark}{Remarks}

\newcommand{\id}{\mathbb{I}}
\newcommand{\cD}{\mathcal{D}}
\newcommand{\cE}{\mathcal{E}}

\newcommand{\DKL}{D_{\mathrm{KL}}}
\newcommand{\DBS}{D_{\mathrm{BS}}}
\newcommand{\dd}{\mathrm{d}}
\newcommand{\ket}[1]{\lvert #1\rangle}
\newcommand{\ketbra}[1]{\lvert #1\rangle\!\langle #1\rvert}
\newcommand{\conv}{\operatorname{conv}}

\newcommand{\diag}{\operatorname{diag}}
\newcommand{\linspan}{\operatorname{span}}
\newcommand{\lamax}{\lambda_{\max}}

\title{Channel-Constrained Information Scheduling for Quantum Diffusion}

\author{
Qipeng Qian$^{1,2}$ \qquad Yuntao Qian$^{2}$\\[4pt]
$^{1}$SUPCON Technology, Hangzhou, China\\
$^{2}$College of Artificial Intelligence, Zhejiang University, Hangzhou 310027, China\\[4pt]
\texttt{qianqipeng@supcon.com} \qquad
\texttt{ytqian@zju.edu.cn}
}

\date{}

\begin{document}
\maketitle

\begin{abstract}
Information-based diffusion schedules organize the reverse problem by prescribing how much source information remains after each forward step. For mixed-state quantum diffusion, density-level targets need not correspond to trajectories reachable from the current realization. We introduce \emph{channel-constrained information scheduling}: each step optimizes over source-blind pure-state refinements reachable from the current state-resolved record through the prescribed channel. This defines a channel-constrained information clock $V_\xi$, whose inverse selects the least-corrupting step that reaches a target information level. Recursive inversion constructs a realizable equal-information schedule minimizing the worst source-information burden in Bayes reverse prediction. We establish attainment, continuity, monotonicity, convexity, finite support, and exact primal/dual formulations for $V_\xi$. We further give a sharp structural characterization: binary-qubit depolarizing trajectories attain the density-level Belavkin--Staszewski envelope, whereas every $d\ge3$ admits binary noncommuting families with a strict static--dynamic separation. Qutrit experiments show that the proposed scheduler nearly equalizes the realized multi-step burden, and learned reverse predictors recover the predicted allocation.
\end{abstract}

\section{Introduction}
A diffusion schedule determines how quickly a forward process removes information and, consequently, how the reverse problem is distributed across time. Classical information-based schedules exploit this observation by reparameterizing a fixed corruption path using entropy, mutual information, denoising difficulty, or related quantities \citep{stancevic2025entropic,raya2026infonoise,jeon2025itdd,ambrogioni2026atoms}. This viewpoint rests on a usually implicit premise: the intermediate representations selected by the information clock must themselves lie on a trajectory realizable by the prescribed forward dynamics.

Mixed-state quantum diffusion makes that premise nontrivial. A density operator admits many pure-state realizations that can retain different amounts of source information, so a density-level optimum need not be compatible with the trajectory already realized. We make reachability part of the scheduling objective itself: from the current trajectory, choose the least-corrupting channel step and source-blind refinement that attain the next information target, then iterate this construction over the full path.

This issue arises directly in trajectory-based quantum generative models. QuDDPM introduced stepwise generative learning of quantum-state distributions \citep{zhang2024quddpm}, and MSQuDDPM uses a mixed-state depolarizing forward process with a prescribed schedule \citep{kwun2024msquddpm}. Recent approaches further expose stochastic pure trajectories, monitored records, manifold-valued scores, and reverse quantum channels \citep{liu2025measurementqdm,xu2026ssdm,bompais2026reverse,gabbassov2025pauli}. In such models, the state-resolved forward record is part of the trajectory itself. A schedule should therefore be defined on representations that can be propagated from that record, not on density-level optima chosen independently at each time.

We solve this constructive problem with a channel-constrained information clock. Let $K$ be the current state-resolved record and $X$ the finite source index whose distinguishability is progressively erased. After a depolarizing step with retention $a$, each realized pure branch becomes mixed and is refined into a new pure-state record using $K$ but not $X$. Among all such source-blind, channel-compatible refinements, we define $V_\xi(a)$ as the minimum remaining $I(X{:}K')$. Unlike a density-only information scalar, $V_\xi$ depends on the channel and on the trajectory representation $\xi$ from which the next step must be generated.

The key methodological step is to \emph{invert} this reachable information value. Along the induced Markov chain $X\to K_{t-1}\to K_t$, the source-information contribution to Bayes reverse prediction is
\begin{equation}
 c_t:=\mathbb E\,\DKL\!\left(P(K_{t-1}|K_t,X)\middle\|P(K_{t-1}|K_t)\right)
 =I(X{:}K_{t-1})-I(X{:}K_t).
\label{eq:intro-reverse-kl}
\end{equation}
We prove that $V_\xi$ is attained, continuous, monotone, and convex. Its inverse therefore selects the largest retention---equivalently, the least-corrupting forward step---whose reachable optimum meets a prescribed next information level. Recursive inversion realizes equal decrements $c_t=J_0/T$ on an actual trajectory; because the decrements telescope to a fixed total, this construction is minimax for the hardest reverse-information step and optimal for every convex separable imbalance penalty. Thus $V_\xi$ is the reachable coordinate used to construct the schedule itself.

To determine when this channel constraint changes the schedule, we compare $V_\xi$ with the density-level Belavkin--Staszewski (BS) envelope
\begin{equation}
 B(\cE)=\sum_x p_x\DBS(\rho_x\|\bar\rho),\qquad \bar\rho=\sum_xp_x\rho_x.
\label{eq:intro-B}
\end{equation}
Optimal-unraveling theory identifies $B$ with a natural static lower envelope \citep{ruibal2025unraveling,matsumoto2013fdiv}. For binary qubits, we construct a BS-optimal trajectory that remains closed under depolarization, so the channel-constrained and density-level clocks coincide along the trajectory. For every $d\ge3$, by contrast, we exhibit a binary noncommuting family for which no finite BS optimum at one interior noise level can be propagated into a BS optimum at another. This sharp static--dynamic separation identifies when density-only scheduling misallocates the realized reverse-information burden. In noncommuting qutrits, recursive channel-constrained scheduling restores an almost uniform burden profile, and reverse predictors trained only on sampled trajectory records recover the same ordering.

\paragraph{Contributions.}
\begin{itemize}[leftmargin=*]
\item \textbf{Channel-constrained information scheduling.} We define the reachable information clock $V_\xi$ from the current state-resolved trajectory, connect its decrement exactly to the source-information term in Bayes reverse prediction, and show that recursive inversion constructs least-corrupting equal-information steps with a minimax worst-step burden.
\item \textbf{Constructive optimization and certification.} We prove attainment, continuity, monotonicity, convexity, and finite support for $V_\xi$, and derive an exact finite moment formulation together with a no-gap operator dual. The primal produces explicit source-blind refinements, while the dual provides lower certificates.
\item \textbf{Sharp static--dynamic characterization and validation.} BS information is channel-exact along binary-qubit depolarizing trajectories, whereas every $d\ge3$ admits a two-state noncommuting family with a strict gap for all finite BS minimizers. Qutrit experiments confirm the resulting multi-step correction and its recovery by learned reverse predictors.
\end{itemize}
\section{Related work}
Classical diffusion research has developed learned schedules and information-based time parameterizations for a corruption path that is taken as given \citep{sohldickstein2015deep,ho2020ddpm,song2021sde,nichol2021improved,kingma2021vdm,karras2022edm,stancevic2025entropic,raya2026infonoise,jeon2025itdd,ambrogioni2026atoms}. Our method changes the optimization domain of scheduling itself: rather than first selecting density-level optima and then checking whether they can be connected, each step optimizes over refinements reachable from the representation already produced by the forward process, and the resulting reachable value is inverted recursively to construct the schedule. Quantum diffusion models now span ensemble, mixed-state, measurement-trajectory, manifold-score, and reverse-channel formulations \citep{chen2024qgdm,zhu2024structure,zhang2024quddpm,kwun2024msquddpm,liu2025measurementqdm,xu2026ssdm,tang2025quadim,bompais2026reverse,gabbassov2025pauli}. CCMQD, for example, constrains learned maps to be CPTP \citep{zhu2025ccmqd}; here the forward map is prescribed, and the optimization is over trajectory realizations compatible with that map.

The closest quantum-information ingredients are HJW decompositions, physically realizable ensembles (PREs), and BS relative entropy. HJW characterizes static decompositions, while PRE theory asks which ensembles can be maintained or realized under monitored open-system dynamics \citep{hughston1993classification,wiseman2001pre,karasik2011bits,warszawski2018pre}. We use reachability as a \emph{scheduling constraint}: the next information value is optimized over the reachable refinement set and then inverted to choose the forward step. BS relative entropy is a maximal quantum $f$-divergence/reverse-test quantity \citep{belavkin1982cstar,matsumoto2010reverse,matsumoto2013fdiv,hiai2017different,bluhm2020bsdpi}. Most directly, \citet{ruibal2025unraveling} identify BS entropy with minimum realization KL and study how initially optimal measures evolve under common dynamics. That static minimum is our density-level benchmark; $V_\xi$ conditions the optimization on the \emph{current realized representation}, re-optimizes over source-blind channel-compatible refinements, and supplies the value function inverted by the scheduler. The qubit closure and higher-dimensional obstruction give the sharp geometry of this scheduling constraint.
\section{Trajectory model and channel-compatible refinements}
\label{sec:setup}
Let $P(X=x)=p_x>0$ and represent each source state as
\begin{equation}
 \rho_x=\sum_{k=1}^m q(k|x)P_k,\qquad J:=I(X{:}K),
\label{eq:trajectory-representation}
\end{equation}
where the distinct rank-one projectors $P_k$ are the state-resolved values of $K$. The index $X$ names the finite source ensemble being tracked; it may represent a condition, mixture component, or empirical sample identity. Thus $J$ measures how much source identity remains in the realized pure-state trajectory, even when different trajectories share the same density operators. The forward channel is
\begin{equation}
 \cD_a(A)=aA+(1-a)\Tr(A)\frac{\id}{d},\qquad a\in[0,1],
\label{eq:depolarizing}
\end{equation}
with multiplicative composition of retentions.

\begin{definition}[Channel-compatible pure-state refinement]
\label{def:causal-refinement}
For a CPTP map $\Phi$, a refinement is a family of pure-state laws $\nu_k$ satisfying
\begin{equation}
 \int Q\,\nu_k(\dd Q)=\Phi(P_k).
\label{eq:barycenter-constraint}
\end{equation}
The law may depend on the observed branch $k$ but not on $X$, so drawing $K'\sim\nu_K$ yields $X\to K\to K'$ and averages exactly to $\Phi(\rho_x)$.
\end{definition}
By HJW, every finite decomposition of $\Phi(P_k)$ can be implemented by measuring an environment in a Stinespring dilation \citep{hughston1993classification}. The refinement may therefore adapt to the realized record $k$ while remaining blind to $X$; the finite-support result below shows that finite branchwise decompositions suffice.

\paragraph{Operational setting.}
Here $K$ is the classical branch record of a trajectory-resolved or monitored forward process. The refinement may depend on the observed branch $K=k$ but not on the hidden source $X$, giving $X\to K\to K'$. We fix the average CPTP map $\Phi$ while allowing the HJW refinement of each known branch output to depend on $k$. This matches settings with an available state-resolved record and selectable unraveling; a fixed microscopic instrument corresponds to a further restriction of the feasible set.

\begin{proposition}[Reverse information burden]
\label{prop:reverse-burden}
For every label-independent step, with $R_t^x=P(K_{t-1}|K_t,X=x)$ and $R_t=P(K_{t-1}|K_t)$,
\begin{align}
 c_t&:=\mathbb E_{X,K_t}\DKL\!\left(R_t^X(\cdot|K_t)\|R_t(\cdot|K_t)\right)
 \label{eq:ct-def}\\
 &=I(X{:}K_{t-1}|K_t)=I(X{:}K_{t-1})-I(X{:}K_t).
\label{eq:ct-info-drop}
\end{align}
\end{proposition}
For population cross-entropy reverse training, the same quantity appears as an explicit additive term:
\begin{equation}
 \mathcal L_t(\theta)
 =H(K_{t-1}|K_t,X)+c_t
 +\mathbb E_{K_t}\DKL\!\left(R_t(\cdot|K_t)\|Q_{\theta,t}(\cdot|K_t)\right).
\end{equation}
The first term is irreducible forward stochasticity, the last is model mismatch, and $c_t$ is the ambiguity created by hiding $X$. Unlike architecture-dependent optimization effects, $c_t$ is fixed once the forward trajectory is fixed. It is therefore a schedule-controlled component of Bayes reverse prediction for state-resolved trajectory models. A density-level clock can misallocate this burden when its intermediate optima are not causally realizable; \cref{sec:numerics} demonstrates the effect in qutrit trajectories. Appendix~\ref{app:reverse} gives the derivation.

\section{Channel-constrained information scheduling}
\label{sec:physical-clock}
\begin{definition}[Channel-constrained information value]
\label{def:V}
\mbox{}\par\noindent For the current representation $\xi=(p_x,q(k|x),P_k)$,
\begin{equation}
 V_\xi(a):=\inf_{\{\nu_k\}} I(X{:}K')
 \quad\text{s.t.}\quad
 \int Q\,\nu_k(\dd Q)=\cD_a(P_k)\ \forall k.
\label{eq:V-def}
\end{equation}
Duplicate labels producing the same projector can be merged without increasing information, so $K'$ is taken to be state-resolved.
\end{definition}
$V_\xi$ restricts the static information objective to realizations reachable from the current branches and is therefore representation-dependent, source-blind, and record-adaptive.

\begin{theorem}[Structure of the channel-constrained information clock]
\label{thm:V-properties}
For every finite trajectory representation $\xi$:
\begin{enumerate}[label=(\roman*),leftmargin=*]
\item the infimum in \eqref{eq:V-def} is attained;
\item $V_\xi(a)$ is nondecreasing, convex, and continuous on $[0,1]$;
\item the endpoints satisfy
\begin{equation}
 V_\xi(0)=0,\qquad V_\xi(1)=J;
\label{eq:V-endpoints}
\end{equation}
\item the keep-or-reset refinement gives
\begin{equation}
 0\le V_\xi(a)\le aJ.
\label{eq:V-linear-upper}
\end{equation}
\end{enumerate}
Moreover, an optimum exists with at most $md^2+2$ output atoms.
\end{theorem}
The proof combines compactness of the pure-state space, lower semicontinuity of mutual information, composition of depolarizing refinements, and Carath\'eodory.  Monotonicity follows because stronger depolarization is a post-processing of weaker depolarization; convexity follows by mixing refinements.  The upper bound is constructive: retain $P_k$ with probability $a$ and otherwise use a source-independent decomposition of $\id/d$ (Appendix~\ref{app:V-properties}).

For a target retained information $0\le\tau\le J$, define
\begin{equation}
 a_\xi^\star(\tau):=\max\{a\in[0,1]:V_\xi(a)\le\tau\}.
\label{eq:a-star}
\end{equation}
\begin{theorem}[Least-noise information step]
\label{thm:least-noise-step}
For $0\le\tau<J$,
\begin{equation}
 V_\xi(a_\xi^\star(\tau))=\tau,
\label{eq:a-star-equality}
\end{equation}
and no larger retention can attain information at most $\tau$.
\end{theorem}
The maximum in \eqref{eq:a-star} is the least-corrupting realization of the requested information drop because increasing $a$ always retains more signal.  Continuity rules out a jump below the target at the maximizing point, while optimizer existence supplies the corresponding causal refinement.

\begin{corollary}[Equal reverse-information schedule]
\label{cor:equal-burden}
For any $J_0>0$ and $T\ge1$, recursive inversion realizes
\begin{equation}
 J_t=\left(1-\frac tT\right)J_0,
\label{eq:equal-J-grid}
\end{equation}
by choosing
\begin{equation}
 a_t=a_{\xi_{t-1}}^\star(J_{t-1}-J_0/T).
\end{equation}
Each step uses the largest retention compatible with its target, and
\begin{equation}
 c_t=\frac{J_0}{T}\qquad(t=1,\ldots,T).
\label{eq:equal-ct}
\end{equation}
\end{corollary}

Equal allocation is elementary once the total burden is fixed; the quantum difficulty is realizing the required intermediate levels. \Cref{thm:V-properties,thm:least-noise-step,cor:equal-burden} provide that missing step by constructing an actual causal trajectory with the equal-burden vector.

\begin{theorem}[Optimal allocation of a causally realizable reverse burden]
\label{thm:allocation-optimality}
For any $T$-step label-independent causal trajectory with the same $J_0$ and $I(X{:}K_T)=0$,
\begin{equation}
 \max_t c_t\ge\frac{J_0}{T}.
\label{eq:minimax-burden-lower}
\end{equation}
More generally, for every convex $\phi:[0,J_0]\to\mathbb R$,
\begin{equation}
 \sum_{t=1}^T\phi(c_t)\ge T\phi(J_0/T).
\label{eq:convex-burden-opt}
\end{equation}
The schedule in \cref{cor:equal-burden} attains both bounds; for strictly convex $\phi$, equality requires the equal-burden vector.
\end{theorem}
The inequalities themselves follow from $\sum_t c_t=J_0$ and Jensen's inequality (Appendix~\ref{app:V-properties}). Their relevance here is that the optimal vector is actually attainable: $V_\xi$ determines how to realize it while retaining as much signal as possible at every step. The resulting schedule is defined by reachable information targets rather than by a retention curve chosen in parameter space.

\section{Static benchmark and channel exactness}
\label{sec:BS}
For faithful $\sigma>0$ and $\rho\ge0$, let
\begin{equation}
 A_{\rho|\sigma}=\sigma^{-1/2}\rho\sigma^{-1/2},\qquad
 \DBS(\rho\|\sigma)=\Tr[\sigma A_{\rho|\sigma}\log A_{\rho|\sigma}],
\label{eq:DBS}
\end{equation}
with $0\log0=0$, and define $B(\cE)=\sum_xp_x\DBS(\rho_x\|\bar\rho)$ when the average $\bar\rho$ is faithful.  Maximal-$f$ data processing yields the following density-only envelope.

\begin{proposition}[BS lower envelope]
\label{prop:BS-lower}
Every finite trajectory representation of $\cE$ satisfies $B(\cE)\le I(X{:}K)$; consequently after a causal step,
\begin{equation}
 B(\cE')\le V_\xi(a).
\label{eq:B-lower-V}
\end{equation}
\end{proposition}
The inequality follows by viewing any finite trajectory representation as a classical-to-quantum preparation channel.  For $f(t)=t\log t$, operator Jensen/data processing maps the classical $\DKL(\mu_x\|\mu)$ to the BS divergence of the corresponding density operators, giving $\DBS(\rho_x\|\bar\rho)\le\DKL(\mu_x\|\mu)$ term by term.  Hence $B$ is a static lower envelope: it optimizes the target representation without accounting for the trajectory from which that representation must be reached.

For binary ensembles the static bound is tight.
\begin{proposition}[Binary static saturation]
\label{prop:binary-saturation}
For $\cE=\{p,\rho_0;1-p,\rho_1\}$ with $0<p<1$ and faithful average,
\begin{equation}
 \inf I(X{:}K)=B(\cE).
\label{eq:binary-static}
\end{equation}
\end{proposition}
Indeed $A_x=\bar\rho^{-1/2}\rho_x\bar\rho^{-1/2}$ obey $pA_0+(1-p)A_1=\id$, hence $A_0$ and $A_1$ commute.  If $A_x=\sum_j a_{xj}\ketbra{u_j}$ in a shared eigenbasis, then $w_j=\langle u_j|\bar\rho|u_j\rangle$ and $\ket{\phi_j}=\bar\rho^{1/2}\ket{u_j}/\sqrt{w_j}$ give coefficients $q(j|x)=w_ja_{xj}$ with $\rho_x=\sum_jq(j|x)\ketbra{\phi_j}$ and $I(X{:}J)=B(\cE)$.  Appendix~\ref{app:binary-static} gives the full calculation.

Binary saturation solves each noise level in isolation. Scheduling imposes one additional requirement: the target optimizer must be generated by a label-independent transition from the optimizer already realized. The problem is therefore one of compatibility between optimal representations at successive levels, not of static decomposition alone.

The following criterion tests that compatibility for the $d$-atom BS representations above, and more generally whenever the selected target projectors are linearly independent.
\begin{theorem}[Moving-basis criterion]
\label{thm:moving-basis}
Let $\{P_j\}_{j=1}^{d}$ and $\{Q_\ell\}_{\ell=1}^{d}$ be selected source and target BS representations, and assume that the target projectors $\{Q_\ell\}_{\ell=1}^{d}$ are linearly independent (in particular, this holds for the basis-generated representation of \cref{prop:binary-saturation}).  A CPTP map $\Phi$ propagates the selected source representation into the selected target representation iff
\begin{equation}
 \Phi(P_j)\in\conv\{Q_1,\ldots,Q_d\}\quad\forall j,
\label{eq:simplex-closure}
\end{equation}
equivalently
\begin{equation}
 \Phi(P_j)=\sum_\ell W_{\ell j}Q_\ell
\label{eq:W-closure}
\end{equation}
for a column-stochastic $W$.
\end{theorem}
Necessity is immediate from the conditional branch probabilities.  For sufficiency, the convex decomposition of each $\Phi(P_j)$ is realizable by HJW because the source branch $j$ is already known.  Linear independence of the target projectors makes their coefficients unique, so the same column-stochastic $W$ uniquely propagates the ensemble coefficients for both source values.  The criterion therefore converts dynamic attainability into a geometric simplex-closure test.  Arbitrary, possibly degenerate, BS minimizers are handled by the equality-support argument below rather than by assuming a unique common basis.

\begin{theorem}[Exact binary-qubit information clock]
\label{thm:qubit-exact}
For every binary qubit ensemble with prior $0<p<1$, distinct states, and $0\le\lambda'<\lambda\le1$, let $\xi_\lambda$ denote the two-atom BS-optimal trajectory supported on the chord endpoints of the depolarized pair at cumulative retention $\lambda$.  Then $\xi_\lambda$ propagates under $\cD_{\lambda'/\lambda}$ to the corresponding chord-endpoint BS-optimal trajectory $\xi_{\lambda'}$.  Writing the affine Bloch line as $r=c+sv$ with $c\perp v$ and
\begin{equation}
 n_\pm(\lambda)=\lambda c\pm h_\lambda v,\qquad h_\lambda=\sqrt{1-\lambda^2\|c\|^2},
\label{eq:chord-endpoints}
\end{equation}
the transition is binary symmetric with
\begin{equation}
 W_{\lambda\to\lambda'}=\frac12\begin{pmatrix}1+\eta&1-\eta\\1-\eta&1+\eta\end{pmatrix},\qquad
 \eta=\frac{\lambda'}{\lambda}\frac{h_\lambda}{h_{\lambda'}}\in[0,1],
\label{eq:qubit-W}
\end{equation}
and
\begin{equation}
 V_{\xi_\lambda}(\lambda'/\lambda)=B(\cE_{\lambda'}).
\label{eq:V-B-qubit}
\end{equation}
\end{theorem}
In the Bloch ball, the geometry is explicit.  Depolarization scales every Bloch vector, so the noisy pair stays on the scaled affine line $\lambda c+sv$.  The only pure states on that line are its two sphere intersections $n_\pm(\lambda)$.  After a step $a=\lambda'/\lambda$, each old endpoint lies on the target chord with the binary-symmetric weights above; checking $\eta\le1$ reduces exactly to $\lambda'^2\le\lambda^2$.  Thus the static BS optimum is dynamically feasible along the entire qubit path (Appendix~\ref{app:qubit}).

\begin{theorem}[Higher-dimensional obstruction]
\label{thm:highd-nogo}
Let $d\ge3$, let the binary prior satisfy $0<p<1$, and let $\rho_0=\ketbra{\psi_0}$ and $\rho_1=\ketbra{\psi_1}$ be distinct nonorthogonal pure states.  For every $0<\lambda'<\lambda<1$, no BS-optimal finite trajectory of
$\rho_x^\lambda=\lambda\rho_x+(1-\lambda)\id/d$ can be propagated by $\cD_{\lambda'/\lambda}$ into a BS-optimal trajectory at $\lambda'$.  Hence every finite BS-optimal source trajectory satisfies
\begin{equation}
 V_{\xi_\lambda}(\lambda'/\lambda)>B(\cE_{\lambda'}).
\label{eq:highd-strict}
\end{equation}
\end{theorem}
\paragraph{Proof idea.}
Let $S=\linspan\{\psi_0,\psi_1\}$.  On $S^\perp$, the two depolarized states and their average coincide, so the normalized likelihood operator has eigenvalue $1$ there; on $S$ it has two distinct one-dimensional eigenspaces and no eigenvalue $1$.  This spectral split constrains \emph{every} BS minimizer, not only the canonical common-basis construction.  Equality in the BS lower bound makes operator Jensen tight for $f(t)=t\log t$; strict convexity then forces every finite minimizing atom, after normalization by the average state, to lie inside one eigenspace of the normalized likelihood operator.  Hence every BS-optimal trajectory is supported only on the two allowed directions in $S$ together with $S^\perp$.

Because the average state is strictly positive on $S^\perp$, a source optimizer must contain a positive-mass atom there.  Depolarizing this branch creates the isotropic block $(1-a)\id_S/d$ on the data subspace.  A target BS optimizer has only two allowed directions in $S$, so reproducing this rank-two isotropic block from those directions forces them to be orthogonal with equal weights.  But then every target atom lies either in one common orthogonal basis of $S$ or in $S^\perp$, which makes both target density operators diagonal in the same block basis and hence commuting.  This contradicts
$[\rho_0^{\lambda'},\rho_1^{\lambda'}]=\lambda'^2[\rho_0,\rho_1]\ne0$.  The obstruction is genuinely dynamic: the density-level optimum exists at both endpoints, yet no optimal trajectory at the first endpoint can reach any optimal trajectory at the second.  Appendix~\ref{app:BS-equality} proves the equality-support lemma and Appendix~\ref{app:highd} gives the complete argument.

The BS clock is therefore channel-exact on the binary-qubit chord trajectory and strictly optimistic for the higher-dimensional family. If $\Gamma_t:=I(X{:}K_t)-B(\cE_t)$ denotes the realization gap, then
$c_t=(B_{t-1}-B_t)+(\Gamma_{t-1}-\Gamma_t)$, making explicit how trajectory compatibility modifies density-level equalization. The separation is set-level: the equality-support lemma characterizes every finite BS minimizer, and the forward channel cannot connect the optimal sets across interior noise levels. Structurally, Bloch-chord geometry closes the qubit optimum, whereas the additional orthogonal complement from $d=3$ onward creates an incompatible isotropic block (Appendix~\ref{app:additional}).

\section{Finite optimization and certification}
\label{sec:computation}
Let $\pi_k=P(K=k)$, $\eta_x^k=P(X=x|K=k)$, and
\begin{equation}
 g(\alpha)=\DKL\!\left(\sum_k\alpha_k\eta^k\middle\|p\right),\qquad \alpha\in\Delta_m.
\label{eq:g-alpha}
\end{equation}
\begin{theorem}[Finite-support primal]
\label{thm:finite-primal}
For $M_k(a)=\pi_k\cD_a(P_k)$,
\begin{equation}
 V_\xi(a)=\min_{\tau_\ell,\alpha_\ell,Q_\ell}\sum_\ell\tau_\ell g(\alpha_\ell)
\label{eq:finite-primal}
\end{equation}
subject to
\begin{equation}
 \sum_\ell\tau_\ell\alpha_{\ell k}Q_\ell=M_k(a)\ \forall k,
 \quad \tau_\ell\ge0,\ \alpha_\ell\in\Delta_m,\ Q_\ell\ \text{pure},
\label{eq:finite-primal-constraints}
\end{equation}
and an optimum uses at most $md^2+2$ atoms.
\end{theorem}
The lifted moment problem is equivalent to the original measure optimization and admits a no-gap semi-infinite operator dual (Appendix~\ref{app:dual}).  In compact form, with Hermitian multipliers $Y_k$ and scalar $\beta$, dual feasibility is
\begin{equation}
 \beta+\lambda_{\max}\!\left(\sum_k\alpha_kY_k\right)\le g(\alpha)
 \qquad\forall\alpha\in\Delta_m,
\end{equation}
and every feasible dual point gives the certified lower bound
\begin{equation}
 \beta+\sum_k\Tr[Y_kM_k(a)]\le V_\xi(a).
\end{equation}
Finite-library primal solutions give constructive upper approximations, while feasible dual points certify lower bounds. Schedule construction is a one-dimensional outer inversion in $a$ around a nontrivial pure-state optimization. The primal/dual pair therefore provides a principled route to progressively tighten construction and certification, especially in higher dimension where the BS density-level curve is only a static lower envelope.

\section{Numerical validation}
\label{sec:numerics}
We evaluate the proposed scheduler at three levels: whether the reachable information clock can be constructed accurately, whether recursive inversion changes the realized multi-step reverse-information allocation relative to density-static scheduling, and whether the resulting ordering is visible to learned reverse predictors. The qubit case is an analytically exact control; noncommuting qutrits provide the smallest setting in which the channel constraint is active. Appendix~\ref{app:numerics} contains the complete grids, solver diagnostics, trajectory checks, and training protocol.

For an equal-prior binary qubit with $\ket{\psi_0}=\ket0$, $\ket{\psi_1}=\cos(\pi/3)\ket0+\sin(\pi/3)\ket1$, source retention $0.9$, and nine target retentions, a restricted library containing the exact target BS atoms gives
\begin{equation}
 \max_{\lambda'}|\widehat V-B|=3.53\times10^{-9}\ \text{nats},
\label{eq:num-qubit-gap}
\end{equation}
recovering the analytic clock to numerical precision.

\paragraph{One-step channel-constrained correction.}
For $d=3$, all 12 settings in our qutrit grid fail BS-basis closure. At the representative setting $(\theta,\lambda,r)=(\pi/4,0.85,0.7)$, the independently optimized density-static target has $B_{\rm tgt}=0.141985$ but is not reachable through the prescribed depolarizing step: its maximum source-atom closure residual is $6.62\times10^{-2}$. A restricted primal instead constructs a source-blind channel-compatible refinement with
\begin{equation}
 \widehat V=0.143885,\qquad \widehat V-B_{\rm tgt}=1.90069\times10^{-3},
\label{eq:num-qutrit-example}
\end{equation}
where $\widehat V$ is the information of the explicit finite-library refinement. Its maximum branchwise barycenter residual is $4.07\times10^{-9}$. The corresponding reverse burden is $0.223874$ nats, exactly matching an independent Bayes cross-entropy evaluation of the realized transition. The density-static counterfactual gives $0.225775$ nats. \Cref{tab:qutrit-causal-correction} summarizes the one-step effect; the strict inequality itself is supplied analytically by \cref{thm:highd-nogo}.

\begin{table}[t]
\centering
\caption{One-step channel-constrained correction in the noncommuting qutrit example. The density-static optimum is not closed under the prescribed forward step; the second column is an explicit feasible refinement.}
\label{tab:qutrit-causal-correction}
\small
\begin{tabular}{lrr}
\toprule
Quantity & Density-static & Feasible causal \\
\midrule
Target information & $0.141985$ & $0.143885$ \\
Reverse decrement & $0.225775$ & $0.223874$ \\
Feasibility residual & $6.62\times10^{-2}$ & $4.07\times10^{-9}$ \\
Bayes CE gap & unattainable target & $0.223874$ \\
\bottomrule
\end{tabular}
\end{table}

\paragraph{Multi-step scheduling.}
We then propagate the same noncommuting qutrit family for $T=4$ steps from cumulative retention $\lambda_0=0.85$. The density-static schedule chooses retentions that make $B(\cE_t)$ linear in $t$, while every transition is generated from the representation actually realized at the previous step. The channel-constrained schedule re-optimizes a feasible source-blind refinement at each step to follow a linear grid in $J_t=I(X{:}K_t)$. We repeat the construction with $128$, $256$, and $512$ fixed Haar atoms added to the structured library, independently optimizing the channel-constrained path at each library size.

\begin{table}[t]
\centering
\caption{Four-step qutrit scheduling under independently optimized finite libraries. Density-static scheduling equalizes $B$; channel-constrained scheduling equalizes the realized information trajectory. Lower CV and $\max_t c_t$ are better.}
\label{tab:qutrit-multistep}
\scriptsize
\begin{tabular}{@{}crrrrrr@{}}
\toprule
$N$ & CV$_{\rm static}$ & CV$_{\rm causal}$ & CV red. & $\max c_t^{\rm static}$ & $\max c_t^{\rm causal}$ & max red. \\
\midrule
128 & $0.3076$ & $\mathbf{0.0007}$ & $99.8\%$ & $0.1241$ & $\mathbf{0.0920}$ & $25.8\%$ \\
256 & $0.2473$ & $\mathbf{0.0016}$ & $99.4\%$ & $0.1251$ & $\mathbf{0.0921}$ & $26.3\%$ \\
512 & $0.2404$ & $\mathbf{0.0064}$ & $97.3\%$ & $0.1125$ & $\mathbf{0.0928}$ & $17.5\%$ \\
\bottomrule
\end{tabular}
\end{table}

\Cref{tab:qutrit-multistep} exhibits the scheduling consequence of the realizability gap. Equal decrements of the density-static clock leave a strongly nonuniform realized burden, whereas channel-constrained inversion makes the realized decrements nearly constant and lowers the hardest step at every library size. The optimized retention parameters vary with the finite library, but the information allocation selected by the channel-constrained construction is stable.

\paragraph{Reverse learning on the qutrit trajectories.}
To test whether this geometry is visible to a learned reverse model, we freeze the independently optimized $N=256$ and $N=512$ trajectories and train paired timestep-conditioned predictors $Q_\theta(K_{t-1}\mid K_t,t)$ and $Q_\phi(K_{t-1}\mid K_t,t,X)$ from sampled state-resolved records. Their held-out cross-entropy difference
\begin{equation}
 \widehat c_t:=\mathrm{CE}_{\mathrm{noX},t}-\mathrm{CE}_{\mathrm{withX},t}
\label{eq:learned-burden}
\end{equation}
estimates the exact burden $c_t$. Model selection uses validation cross entropy only; the trajectories, schedules, and exact burdens are fixed before training.

\begin{table}[t]
\centering
\caption{Reverse learning on the four-step qutrit trajectories. Learned quantities are means over ten fixed seeds. ``Dir.'' counts seeds in which channel-constrained scheduling improves both CV and the worst step.}
\label{tab:qutrit-learned}
\scriptsize
\begin{tabular}{@{}crrrrrc@{}}
\toprule
$N$ & exact CV red. & learned CV red. & exact max red. & learned max red. & MAE & Dir. \\
\midrule
256 & $99.4\%$ & $\mathbf{95.2\%}$ & $26.3\%$ & $\mathbf{25.4\%}$ & $9.72\!\times\!10^{-4}$ & $10/10$ \\
512 & $97.3\%$ & $\mathbf{92.0\%}$ & $17.5\%$ & $\mathbf{16.3\%}$ & $1.36\!\times\!10^{-3}$ & $10/10$ \\
\bottomrule
\end{tabular}
\end{table}

The learned predictors reproduce both the flattening and the hardest-step reduction with mean burden MAE $1.17\times10^{-3}$ nats; all ten seeds agree on both comparisons. A learned categorical table and empirical conditional-frequency estimator recover the same ordering (Appendix~\ref{app:numerics}).

Full reverse rollouts reveal a second trajectory effect: error propagation. At $N=512$, the channel-constrained path lowers endpoint trace distance from $0.01253$ to $0.00804$ despite larger local model mismatch; single-step hybrids trace the gain to stronger contraction of the dominant $t=4$ and $t=2$ perturbations ($0.00379$ vs. $0.00691$ endpoint error for isolated $t=4$). At $N=256$, larger local mismatch dominates, yielding $0.00921$ vs. $0.00753$. These regimes separate the exact information-allocation objective from finite-model approximation and multi-step propagation (Appendix~\ref{app:numerics}).

The binary-qubit experiment provides an analytically solvable control where the BS clock is already channel-exact. On an eight-step trajectory, equal-information scheduling attains the exact minimax burden, while linear and cosine retention grids concentrate substantially more burden in their hardest steps. Learned predictors recover the analytic profile to roughly $10^{-3}$ nats; full results are in Appendix~\ref{app:numerics}.
\section{Discussion and conclusion}
We introduced channel-constrained information scheduling for trajectory-resolved quantum diffusion. Its central object, the reachable value $V_\xi$, optimizes each information step over source-blind refinements generated from the current state-resolved record. Inverting $V_\xi$ is the scheduling rule: it realizes prescribed information levels with maximal retention, and a linear information grid gives the minimax allocation of the source-information burden in Bayes reverse prediction.

The static--dynamic theory characterizes the geometry of the method. BS information is a universal density-level lower envelope; binary-qubit depolarizing trajectories attain it along a closed optimal chord trajectory, while every $d\ge3$ admits binary noncommuting families whose BS-optimal sets are not channel-connected. Commuting ensembles remain channel-exact in arbitrary dimension (Appendix~\ref{app:additional}). These results delineate precisely when channel-constrained scheduling coincides with, or departs from, density-only information allocation.

The qutrit experiments realize the intended multi-step correction: recursive inversion nearly equalizes the information decrements, reduces the hardest reverse-information step, and produces the same ordering in learned reverse predictors. Full-rollout diagnostics expose a complementary trajectory effect: finite-model endpoint error depends not only on local reverse approximation but also on how perturbations propagate through the learned chain. Controlled hybrids isolate this propagation sensitivity, clarifying the relation between the exact scheduling objective and learned rollout behavior (Appendix~\ref{app:numerics}).

The definition of $V_\xi$, the reverse-information identity, and the finite moment/dual formulations extend to arbitrary CPTP maps. The inversion theorem uses the ordering, composition, continuity, and fully mixing endpoint of depolarization, providing a direct template for other ordered noise families. Fixed microscopic instruments fit the same framework after restricting the feasible refinement set.

\medskip
\paragraph{Reproducibility.}
The Appendix contains complete proofs, solver diagnostics, qutrit trajectory checks, finite-library sweeps, reverse-learning details, and the analytic qubit control. All restricted-primal values are reported as feasible upper approximations, and the learning experiments use fixed seeds and frozen trajectories.

\subsection*{AI use statement}
Large language models were used as general-purpose research assistance for language polishing, identifying potentially relevant related work, and generating portions of the experimental code. References surfaced with AI assistance were checked by the authors against the original sources before citation. AI-generated code used in the reported experiments was manually reviewed, tested, and checked against the declared experimental protocol before execution. Reported numerical results were obtained by executing the reviewed code rather than by asking a language model to generate result values. The authors take full responsibility for the paper's claims, citations, code, and experimental results.

\bibliographystyle{plainnat}
\bibliography{refs_causal_qdm}

\clearpage
\appendix

\section{Probability-measure conventions and causal refinements}
\label{app:measure}
Let $\mathsf P_d$ denote the compact set of rank-one projectors on the $d$-dimensional Hilbert space, identified with complex projective space.  Probability measures on $\mathsf P_d$ are equipped with the weak topology.  The barycenter map
\begin{equation}
\Lambda(\mu):=\int_{\mathsf P_d}Q\,\mu(\dd Q)
\end{equation}
is continuous because all matrix entries of $Q$ are bounded continuous functions on $\mathsf P_d$.

Given a finite current record $K\in[m]$ with $P(K=k)=\pi_k>0$, a causal refinement is a family $\{\nu_k\}_{k=1}^{m}$ of probability measures on $\mathsf P_d$ such that
\begin{equation}
\Lambda(\nu_k)=\Phi(P_k).
\end{equation}
The joint law is
\begin{equation}
P(X=x,K=k,Q\in A)=p_x q(k|x)\nu_k(A).
\end{equation}
Because the kernel from $K$ to $Q$ does not depend on $X$, the Markov relation $X\to K\to Q$ holds.  Averaging over $k$ gives
\begin{equation}
\mathbb E[Q|X=x]
=\sum_kq(k|x)\Phi(P_k)
=\Phi(\rho_x).
\end{equation}
Thus every feasible refinement is an exact trajectory-level realization of the prescribed forward quantum channel on the current ensemble.

For each known pure input $P_k$, a Stinespring dilation of $\Phi$ yields a purification of $\Phi(P_k)$.  Any finite ensemble decomposition of that output can be generated by a measurement on the purifying environment by the HJW theorem \citep{hughston1993classification}.  Since the controller knows the current state-resolved record $k$, the environment measurement may be selected adaptively as a function of $k$.  The finite-support theorem shows that such finite decompositions suffice for the optimal constructions.  This gives the operational interpretation used throughout the paper: channel compatibility is imposed branch by branch, conditioned only on the state-resolved trajectory information already available.

\section{Proofs for the reverse information burden}
\label{app:reverse}
\begin{proof}[Proof of \cref{prop:reverse-burden}]
The Markov property $X\to K_{t-1}\to K_t$ implies
\begin{equation}
I(X{:}K_t|K_{t-1})=0.
\end{equation}
Applying the chain rule to $I(X{:}K_{t-1},K_t)$ in two orders gives
\begin{align}
I(X{:}K_{t-1},K_t)
&=I(X{:}K_{t-1}),\\
&=I(X{:}K_t)+I(X{:}K_{t-1}|K_t),
\end{align}
which proves the final equality in \eqref{eq:ct-info-drop}.

The conditional mutual information has the standard posterior-divergence form
\begin{align}
I(X{:}K_{t-1}|K_t)
&=\sum_{x,j}P(x,j)
\DKL\!\left(
P(K_{t-1}|x,j)
\middle\|
P(K_{t-1}|j)
\right),
\end{align}
which is \eqref{eq:ct-def}.

For population cross-entropy reverse training, the exact decomposition is
\begin{equation}
 \mathcal L_t(\theta)=H(K_{t-1}|K_t,X)+c_t+\mathbb E_{K_t}\DKL(R_t\|Q_{\theta,t}).
\label{eq:ce-decomp}
\end{equation}
Indeed, conditioning on $K_t=j$ gives
\begin{align}
\mathbb E[-\log Q_{\theta,t}(K_{t-1}|K_t)]
&=H(K_{t-1}|K_t)
+\mathbb E_{K_t}\DKL(R_t\|Q_{\theta,t}).
\end{align}
Finally,
\begin{equation}
H(K_{t-1}|K_t)
=H(K_{t-1}|K_t,X)+I(X{:}K_{t-1}|K_t),
\end{equation}
and substituting the first part of the proof yields \eqref{eq:ce-decomp}.
\end{proof}

\section{Properties of the channel-constrained information value}
\label{app:V-properties}
We prove \cref{thm:V-properties,thm:least-noise-step,cor:equal-burden,thm:allocation-optimality}.  Throughout, $\xi$ has finitely many distinct pure states $P_1,\ldots,P_m$ and source probabilities $\pi_k>0$.

\subsection{Existence of an optimizer}
For each $k$, define
\begin{equation}
\mathcal R_k(a)
:=\{\nu\in\mathcal P(\mathsf P_d):\Lambda(\nu)=\cD_a(P_k)\}.
\end{equation}
The space $\mathcal P(\mathsf P_d)$ is compact in the weak topology because $\mathsf P_d$ is compact.  Each $\mathcal R_k(a)$ is closed under weak convergence by continuity of the barycenter map, hence compact.  Their finite product is compact.

A family $(\nu_1,\ldots,\nu_m)$ induces joint measures $P_{XQ}$ and $P_XP_Q$ continuously under weak convergence.  Mutual information can be written as
\begin{equation}
I(X{:}Q)=\DKL(P_{XQ}\|P_XP_Q).
\end{equation}
Relative entropy is lower semicontinuous under weak convergence on the compact output space with finite $X$.  Therefore $I(X{:}Q)$ is lower semicontinuous on the compact feasible product, and the infimum is attained.

\subsection{Monotonicity}
Take $0\le a_1\le a_2\le1$.  The case $a_2=0$ is trivial.  Otherwise
\begin{equation}
\cD_{a_1}=\cD_{a_1/a_2}\circ\cD_{a_2}.
\end{equation}
Start from any causal refinement of $\cD_{a_2}$.  Conditional on its pure output $Q$, set $b=a_1/a_2$ and apply the explicit measurable keep-or-reset kernel
\begin{equation}
 \kappa_b(\dd Q'|Q)=b\,\delta_Q(\dd Q')+\frac{1-b}{d}\sum_{r=1}^{d}\delta_{E_r}(\dd Q'),
\end{equation}
where $\{E_r\}$ is any fixed rank-one orthonormal-basis decomposition of $\id/d$.  Its barycenter is $\cD_b(Q)$, so the composed refinement is feasible for $a_1$ and is a Markov post-processing of the $a_2$ output.  Data processing therefore gives
\begin{equation}
V_\xi(a_1)\le V_\xi(a_2).
\end{equation}

\subsection{Convexity}
Let $a=\theta a_1+(1-\theta)a_2$ with $0\le\theta\le1$.  For $i\in\{1,2\}$ choose $\varepsilon$-optimal refinements $\nu_k^{(i)}$.  The mixture
\begin{equation}
\nu_k=\theta\nu_k^{(1)}+(1-\theta)\nu_k^{(2)}
\end{equation}
is feasible because
\begin{align}
\Lambda(\nu_k)
&=\theta\cD_{a_1}(P_k)+(1-\theta)\cD_{a_2}(P_k)\\
&=\cD_a(P_k).
\end{align}
For fixed input distribution, mutual information is convex in the conditional output law \citep{cover2006elements}.  Hence
\begin{equation}
V_\xi(a)
\le\theta V_\xi(a_1)+(1-\theta)V_\xi(a_2)+\varepsilon.
\end{equation}
Letting $\varepsilon\downarrow0$ proves convexity.

\subsection{Endpoint values and the keep-or-reset bound}
At $a=0$, every input is mapped to $\id/d$.  Choose the same pure-state decomposition of $\id/d$ for every $k$.  The output is independent of $K$ and therefore of $X$, so $V_\xi(0)=0$.

At $a=1$, the barycenter constraint for branch $k$ is
\begin{equation}
\int Q\,\nu_k(\dd Q)=P_k.
\end{equation}
Since a pure state is an extreme point of the convex state space, $\nu_k=\delta_{P_k}$.  The $P_k$ are distinct, so the output record is equivalent to $K$ and $V_\xi(1)=J$.

If $d=1$, then all pure projectors coincide and $J=0$, so the upper bound is trivial.  Assume henceforth $d\ge2$.  For the upper bound, choose an orthonormal basis $\{\ket{e_r}\}_{r=1}^{d}$ generically so that none of the rank-one projectors $E_r:=\ketbra{e_r}$ equals any of the finitely many $P_k$.  Set
\begin{equation}
\nu_k
=a\delta_{P_k}+\frac{1-a}{d}\sum_{r=1}^{d}\delta_{E_r}.
\label{eq:keep-reset}
\end{equation}
Its barycenter is $\cD_a(P_k)$.  The output support reveals whether the branch came from the first or second term.  With probability $a$ it reveals $K$, while with probability $1-a$ it is independent of $K$.  Therefore
\begin{equation}
I(X{:}Q)=aI(X{:}K)=aJ,
\end{equation}
which proves \eqref{eq:V-linear-upper}.

\subsection{Continuity}
Finite convexity gives continuity on $(0,1)$.  At zero, \eqref{eq:V-linear-upper} gives
\begin{equation}
0\le V_\xi(a)\le aJ\to0.
\end{equation}
For continuity at one, suppose for contradiction that there are $a_n\uparrow1$ and $\epsilon>0$ with
\begin{equation}
V_\xi(a_n)\le J-\epsilon.
\end{equation}
Choose minimizers $\{\nu_k^{(n)}\}$.  Compactness gives a weakly convergent subsequence for every $k$.  The limiting barycenter is $P_k$, so extremality forces the limit to be $\delta_{P_k}$.  The limiting joint law therefore has mutual information $J$.  Lower semicontinuity of mutual information gives
\begin{equation}
J\le\liminf_n V_\xi(a_n)\le J-\epsilon,
\end{equation}
contradiction.

\subsection{Finite support}
The exact moment representation used in \cref{thm:finite-primal} is proved in Appendix~\ref{app:dual}.  Its atom space is compact and each atom is mapped continuously to a vector consisting of $m$ Hermitian $d\times d$ matrices and one objective value.  This vector lives in a real space of dimension at most $md^2+1$.  Carath\'eodory's theorem therefore represents an optimal moment/objective point using at most $md^2+2$ atoms.

\subsection{Proof of the least-noise and equal-burden results}
\begin{proof}[Proof of \cref{thm:least-noise-step}]
The set $\{a:V_\xi(a)\le\tau\}$ is nonempty because it contains $0$, and it is closed by continuity.  Hence its maximum $a_\xi^\star(\tau)$ exists.  Because $\tau<J=V_\xi(1)$, we have $a_\xi^\star(\tau)<1$.  If $V_\xi(a_\xi^\star)<\tau$, continuity would allow a slightly larger $a$ with $V_\xi(a)\le\tau$, contradicting maximality.  Therefore \eqref{eq:a-star-equality} holds.  Existence of an optimizer for $V$ gives a causal refinement with exactly the target information, and any larger retention has $V>\tau$, so no refinement at that retention can meet the target.
\end{proof}

\begin{proof}[Proof of \cref{cor:equal-burden}]
Set $c=J_0/T$.  Starting from $\xi_0$, recursively choose the largest retention satisfying
\begin{equation}
V_{\xi_{t-1}}(a_t)=J_{t-1}-c
\end{equation}
using \cref{thm:least-noise-step}, and choose a minimizing finite causal refinement.  Then $J_t=J_{t-1}-c$, so induction gives \eqref{eq:equal-J-grid}.  Applying \cref{prop:reverse-burden} at every step yields $c_t=J_{t-1}-J_t=c$.
\end{proof}

\begin{proof}[Proof of \cref{thm:allocation-optimality}]
For every label-independent causal step, \cref{prop:reverse-burden} gives
\begin{equation}
c_t=J_{t-1}-J_t\ge0.
\end{equation}
Telescoping over $t=1,\ldots,T$ and using $J_T=0$ yields
\begin{equation}
\sum_{t=1}^{T}c_t=J_0.
\label{eq:burden-telescope-zero}
\end{equation}
Therefore the largest of the $T$ nonnegative burdens is at least their average, proving \eqref{eq:minimax-burden-lower}.  For any convex $\phi$, Jensen's inequality applied to the vector $(c_1,\ldots,c_T)$ and \eqref{eq:burden-telescope-zero} gives
\begin{equation}
\frac1T\sum_{t=1}^{T}\phi(c_t)
\ge
\phi\!\left(\frac1T\sum_{t=1}^{T}c_t\right)
=\phi\!\left(\frac{J_0}{T}\right),
\end{equation}
which is \eqref{eq:convex-burden-opt}.  The schedule of \cref{cor:equal-burden} has $c_t=J_0/T$ for every $t$, so it attains both lower bounds.  Strict convexity gives equality in Jensen only when all arguments coincide, establishing the final claim.
\end{proof}

\section{The BS lower envelope and binary static saturation}
\label{app:binary-static}
\subsection{Proof of the BS lower bound}
It is enough to prove the bound for a finite trajectory representation, because \cref{thm:V-properties} guarantees a finite-support optimizer for $V$.  Let its distinct output projectors be $Q_1,\ldots,Q_L$, with conditional probabilities $\mu_x(\ell)$ and marginal $\mu(\ell)=\sum_xp_x\mu_x(\ell)$.  Remove zero-marginal atoms, so $\mu(\ell)>0$.  Define diagonal classical states
\begin{equation}
R_x:=\diag(\mu_x(1),\ldots,\mu_x(L)),
\qquad
S:=\diag(\mu(1),\ldots,\mu(L)),
\end{equation}
and the preparation channel
\begin{equation}
\Phi(Z):=\sum_{\ell=1}^{L}\langle\ell|Z|\ell\rangle Q_\ell.
\end{equation}
Then $\Phi(R_x)=\rho_x$ and $\Phi(S)=\bar\rho$.  Put $r_x=S^{-1/2}R_xS^{-1/2}$ and define the unital completely positive map
\begin{equation}
T(Z):=\bar\rho^{-1/2}\Phi(S^{1/2}ZS^{1/2})\bar\rho^{-1/2}.
\end{equation}
For $f(t)=t\log t$ with $f(0)=0$, operator Jensen gives $f(T(r_x))\preceq T(f(r_x))$.  Since $T(r_x)=\bar\rho^{-1/2}\rho_x\bar\rho^{-1/2}$,
\begin{align}
\DBS(\rho_x\|\bar\rho)
&=\Tr\!\left[\bar\rho f(T(r_x))\right]\\
&\le \Tr\!\left[\bar\rho T(f(r_x))\right]
=\Tr\!\left[S f(r_x)\right]\\
&=\DKL(\mu_x\|\mu).
\end{align}
Summing with weights $p_x$ yields
\begin{equation}
B(\cE')\le I(X{:}K').
\end{equation}
Applying this inequality to a finite-support optimizer of $V$ proves \eqref{eq:B-lower-V}.

\subsection{Proof of binary saturation}
Let $\bar\rho=p\rho_0+(1-p)\rho_1$ and define
\begin{equation}
A_x=\bar\rho^{-1/2}\rho_x\bar\rho^{-1/2}.
\end{equation}
Then
\begin{equation}
pA_0+(1-p)A_1=\id,
\end{equation}
so
\begin{equation}
A_1=\frac{\id-pA_0}{1-p}
\end{equation}
and $[A_0,A_1]=0$.  Choose a common orthonormal eigenbasis $\{\ket{u_j}\}_{j=1}^{d}$,
\begin{equation}
A_x=\sum_j a_{xj}\ketbra{u_j}.
\end{equation}
Set
\begin{equation}
w_j=\langle u_j|\bar\rho|u_j\rangle>0,
\qquad
\ket{\phi_j}=\frac{\bar\rho^{1/2}\ket{u_j}}{\sqrt{w_j}},
\end{equation}
and
\begin{equation}
q(j|x)=w_ja_{xj}.
\end{equation}
Normalization follows from
\begin{equation}
\sum_jq(j|x)=\Tr(\bar\rho A_x)=\Tr\rho_x=1.
\end{equation}
Moreover,
\begin{equation}
\sum_jq(j|x)\ketbra{\phi_j}
=\bar\rho^{1/2}A_x\bar\rho^{1/2}
=\rho_x.
\end{equation}
The marginal is
\begin{align}
q(j)
&=p q(j|0)+(1-p)q(j|1)\\
&=w_j\bigl(pa_{0j}+(1-p)a_{1j}\bigr)
=w_j.
\end{align}
Hence
\begin{equation}
I(X{:}J)
=\sum_xp_x\sum_j w_ja_{xj}\log a_{xj}.
\label{eq:binary-I-explicit}
\end{equation}
By the definition \eqref{eq:DBS} and the spectral decomposition of $A_x$,
\begin{equation}
\DBS(\rho_x\|\bar\rho)
=\Tr[\bar\rho A_x\log A_x]
=\sum_j w_j a_{xj}\log a_{xj},
\end{equation}
where $0\log0=0$.  Hence \eqref{eq:binary-I-explicit} is exactly $B(\cE)$, proving achievability.  \Cref{prop:BS-lower} gives the matching lower bound.

\subsection{Equality support of a BS-minimizing finite realization}
\label{app:BS-equality}
The higher-dimensional separation requires an equality statement that is slightly more robust than selecting one common-basis minimizer.  We record exactly the support information needed below.

\begin{lemma}[Equality-support lemma]
\label{lem:BS-equality-support}
Let $\rho\ge0$ and $\sigma>0$ be density operators and let $Q_1,\ldots,Q_L$ be rank-one projectors.  Suppose probability vectors $\mu=(\mu_\ell)$ and $\nu=(\nu_\ell)$ with $\nu_\ell>0$ satisfy
\begin{equation}
\rho=\sum_{\ell=1}^{L}\mu_\ell Q_\ell,
\qquad
\sigma=\sum_{\ell=1}^{L}\nu_\ell Q_\ell.
\label{eq:BS-equality-realization}
\end{equation}
Set
\begin{equation}
A:=\sigma^{-1/2}\rho\sigma^{-1/2},
\qquad
r_\ell:=\frac{\mu_\ell}{\nu_\ell}.
\end{equation}
If
\begin{equation}
\DKL(\mu\|\nu)=\DBS(\rho\|\sigma),
\label{eq:BS-equality-assumption}
\end{equation}
then, for every $\ell$ with $Q_\ell=\ketbra{\phi_\ell}$,
\begin{equation}
A\,\sigma^{-1/2}\ket{\phi_\ell}
=r_\ell\,\sigma^{-1/2}\ket{\phi_\ell}.
\label{eq:BS-equality-eigenvector}
\end{equation}
Thus every atom of a finite BS-minimizing realization is supported, after the canonical $\sigma^{-1/2}$ normalization, inside an eigenspace of the normalized likelihood operator $A$.  Repeated eigenvalues may allow arbitrary mixing within the corresponding eigenspace; no nondegeneracy assumption is made.
\end{lemma}

\begin{proof}
Let $E_\ell=\ketbra{\ell}$ on $\mathbb C^L$ and set
\begin{equation}
R:=\sum_\ell\mu_\ell E_\ell,
\qquad
S:=\sum_\ell\nu_\ell E_\ell,
\qquad
r:=S^{-1/2}RS^{-1/2}=\sum_\ell r_\ell E_\ell.
\end{equation}
Define the preparation channel
\begin{equation}
\Phi(Z):=\sum_\ell\langle\ell|Z|\ell\rangle Q_\ell
\end{equation}
and the normalized map
\begin{equation}
T(Z):=\sigma^{-1/2}\Phi(S^{1/2}ZS^{1/2})\sigma^{-1/2}.
\end{equation}
Because $\Phi(S)=\sigma$, $T$ is unital and completely positive, and $T(r)=A$.  Let $f(t)=t\log t$ with $f(0)=0$.  The function $f$ is operator convex on $[0,\infty)$, so operator Jensen gives
\begin{equation}
f(A)=f(T(r))\preceq T(f(r)).
\label{eq:BS-Jensen}
\end{equation}
Moreover,
\begin{align}
\Tr[\sigma T(f(r))]
&=\Tr[S f(r)]
=\sum_\ell\nu_\ell f(r_\ell)
=\DKL(\mu\|\nu),\\
\Tr[\sigma f(A)]
&=\DBS(\rho\|\sigma).
\end{align}
Under \eqref{eq:BS-equality-assumption}, the positive operator $T(f(r))-f(A)$ has zero expectation against the faithful state $\sigma$, hence it is zero:
\begin{equation}
T(f(r))=f(A).
\label{eq:BS-Jensen-equality}
\end{equation}

Let $v$ be a unit eigenvector of $A$ with eigenvalue $a$, and define
\begin{equation}
m_\ell:=\langle v|T(E_\ell)|v\rangle\ge0.
\end{equation}
Since $\sum_\ell T(E_\ell)=T(\id)=\id$, the $m_\ell$ form a probability vector.  Using $A=T(r)$ and \eqref{eq:BS-Jensen-equality},
\begin{equation}
a=\sum_\ell m_\ell r_\ell,
\qquad
f(a)=\sum_\ell m_\ell f(r_\ell).
\end{equation}
The scalar function $f(t)=t\log t$ is strictly convex on $[0,\infty)$, so equality in scalar Jensen implies
\begin{equation}
m_\ell>0\quad\Longrightarrow\quad r_\ell=a.
\label{eq:BS-scalar-equality}
\end{equation}
Fix $\ell$.  For every eigenspace of $A$ whose eigenvalue $a$ differs from $r_\ell$, \eqref{eq:BS-scalar-equality} gives $\langle v|T(E_\ell)|v\rangle=0$ for every vector $v$ in that eigenspace.  Positivity of $T(E_\ell)$ then implies that $T(E_\ell)$ annihilates all such eigenspaces.  Hence the range of $T(E_\ell)$ is contained in the eigenspace of $A$ with eigenvalue $r_\ell$.  Finally,
\begin{equation}
T(E_\ell)=\nu_\ell\,\sigma^{-1/2}Q_\ell\sigma^{-1/2}
\end{equation}
is a nonzero rank-one positive operator with range $\linspan\{\sigma^{-1/2}\ket{\phi_\ell}\}$.  Therefore \eqref{eq:BS-equality-eigenvector} follows.
\end{proof}

\section{Proof of the moving-basis criterion}
\label{app:moving-basis}
We first record a simple linear-independence fact.

\begin{lemma}[Projectors from a basis are linearly independent]
\label{lem:projectors-independent}
If $\{\ket{\phi_j}\}_{j=1}^{d}$ is a basis of the Hilbert space, then the rank-one projectors $Q_j=\ketbra{\phi_j}$ are linearly independent as Hermitian operators.
\end{lemma}
\begin{proof}
Let $\{\ket{\widetilde\phi_j}\}$ be the dual basis, $\langle\widetilde\phi_i|\phi_j\rangle=\delta_{ij}$.  If $\sum_jc_jQ_j=0$, sandwiching by $\ket{\widetilde\phi_i}$ gives $c_i=0$ for every $i$.
\end{proof}

\begin{proof}[Proof of \cref{thm:moving-basis}]
If a causal refinement sends source branch $j$ into the target basis, its conditional branch probabilities give nonnegative weights $W_{\ell j}$ summing to one and satisfying \eqref{eq:W-closure}.  This proves necessity.

Conversely, suppose \eqref{eq:W-closure} holds.  For each known source branch $j$, the right-hand side is a pure-state ensemble decomposition of $\Phi(P_j)$.  By HJW, an environment measurement on a Stinespring dilation can realize this decomposition; because $j$ is in the past trajectory record, the measurement may depend on $j$ without accessing $X$.  Thus $W$ is a causal refinement.

Let $q_x$ be the source coefficient vector in the source common basis.  The channel image has decomposition
\begin{equation}
\Phi(\rho_x)=\sum_\ell (Wq_x)_\ell Q_\ell.
\end{equation}
The target BS representation is another decomposition of the same state using the $d$ target projectors.  By \cref{lem:projectors-independent}, these coefficients are unique.  Hence the target coefficient vector is exactly $Wq_x$ for both labels.  Therefore the entire target BS representation, not only each branch barycenter, is propagated by the same label-independent kernel.
\end{proof}

\section{Binary qubits: exact closure of the BS-optimal trajectory}
\label{app:qubit}
We prove \cref{thm:qubit-exact}.  For two distinct qubit density operators with prior $p\in(0,1)$, the average state is faithful: a zero eigenvector of the average would have zero expectation under both positive states, forcing both states to be supported on the same one-dimensional subspace and hence to coincide.  Thus \cref{prop:binary-saturation} applies also at $\lambda=1$; for $\lambda<1$ faithfulness follows directly from depolarization.  Identify a qubit state with its Bloch vector $r\in\mathbb R^3$, $\|r\|\le1$.  Two distinct states determine a unique affine line.  Write it as
\begin{equation}
L=\{c+sv:s\in\mathbb R\},
\end{equation}
where $\|v\|=1$ and $c\perp v$.  Depolarization scales every Bloch vector by $\lambda$, so the line at cumulative retention $\lambda$ is
\begin{equation}
L_\lambda=\{\lambda c+sv:s\in\mathbb R\}.
\end{equation}
Its intersections with the unit sphere are the two points in \eqref{eq:chord-endpoints}.

For every $\lambda>0$, the two noisy density operators remain distinct.  Any two-state common-basis representation achieving \cref{prop:binary-saturation} has a line segment containing both states; hence its two pure states lie on their unique affine line $L_\lambda$.  The only pure states on that line are the sphere intersections $n_\pm(\lambda)$, so $P_\pm(\lambda)$ are the chord-endpoint BS-optimal projectors.  At $\lambda=0$ the two density operators both equal $\id/2$ and the optimum is nonunique; the inherited endpoints $P_\pm(0)$ nevertheless give a valid BS-optimal representation with equal weights and zero mutual information.  We select this representation when the target level is zero.

Now let $a=\lambda'/\lambda$.  Applying the one-step depolarizing map to the source endpoints gives
\begin{equation}
a n_\pm(\lambda)
=\lambda'c\pm a h_\lambda v.
\end{equation}
The target chord endpoints are
\begin{equation}
n_\pm(\lambda')=\lambda'c\pm h_{\lambda'}v.
\end{equation}
Define
\begin{equation}
\eta=\frac{a h_\lambda}{h_{\lambda'}}
=\frac{\lambda'}{\lambda}\frac{h_\lambda}{h_{\lambda'}}.
\end{equation}
Then
\begin{equation}
a n_+(\lambda)
=\frac{1+\eta}{2}n_+(\lambda')
+\frac{1-\eta}{2}n_-(\lambda'),
\end{equation}
and the analogous equation with $+$ and $-$ interchanged holds for the other endpoint.  It remains to verify $0\le\eta\le1$.  Nonnegativity is immediate.  Squaring the upper bound gives
\begin{align}
\lambda'^2\bigl(1-\lambda^2\|c\|^2\bigr)
&\le\lambda^2\bigl(1-\lambda'^2\|c\|^2\bigr),
\end{align}
which reduces to $\lambda'^2\le\lambda^2$.  Thus \eqref{eq:qubit-W} is a valid stochastic matrix and \cref{thm:moving-basis} proves causal BS closure.

Finally, the resulting output representation is BS-optimal, so its information is $B(\cE_{\lambda'})$ by \cref{prop:binary-saturation}.  The BS lower bound \eqref{eq:B-lower-V} then shows that this feasible value is the minimum $V$, proving \eqref{eq:V-B-qubit}.

\section{Higher-dimensional obstruction}
\label{app:highd}
We prove \cref{thm:highd-nogo}.  Let
\begin{equation}
S=\linspan\{\psi_0,\psi_1\},
\qquad \dim S=2,
\end{equation}
and let $S^\perp$ be its orthogonal complement.  At an interior retention $\lambda$,
\begin{equation}
\rho_x^\lambda
=\lambda\ketbra{\psi_x}+(1-\lambda)\frac{\id}{d}.
\end{equation}
Both states, and therefore their average $\bar\rho^\lambda$, are block diagonal with respect to $S\oplus S^\perp$.  On $S^\perp$,
\begin{equation}
\rho_0^\lambda
=\rho_1^\lambda
=\bar\rho^\lambda
=\frac{1-\lambda}{d}\id_{S^\perp}.
\end{equation}
Hence the normalized likelihood operator
\begin{equation}
A_0^\lambda
=(\bar\rho^\lambda)^{-1/2}
\rho_0^\lambda
(\bar\rho^\lambda)^{-1/2}
\end{equation}
is exactly the identity on $S^\perp$.

On $S$, the eigenvalue $1$ is absent.  Indeed,
\begin{equation}
A_0^\lambda v=v
\end{equation}
for a nonzero $v\in S$ would imply
\begin{equation}
(\rho_0^\lambda-\bar\rho^\lambda)(\bar\rho^\lambda)^{-1/2}v=0,
\end{equation}
which is equivalent, up to the nonzero factor $\lambda(1-p)$, to
\begin{equation}
(\ketbra{\psi_0}-\ketbra{\psi_1})w=0
\end{equation}
for a nonzero $w\in S$.  For distinct nonorthogonal pure states, the difference of these rank-one projectors has eigenvalues
\begin{equation}
\pm\sqrt{1-|\langle\psi_0|\psi_1\rangle|^2}
\end{equation}
on $S$, so it is invertible there.

The two eigenvalues of $A_0^\lambda|_S$ are also distinct.  Otherwise $A_0^\lambda|_S=c\id_S$ for some $c$.  Since
\begin{equation}
pA_0^\lambda+(1-p)A_1^\lambda=\id,
\end{equation}
$A_1^\lambda|_S$ would also be scalar.  Consequently
\begin{equation}
\rho_x^\lambda|_S
=(\bar\rho^\lambda)^{1/2}A_x^\lambda(\bar\rho^\lambda)^{1/2}|_S
=c_x\bar\rho^\lambda|_S,
\end{equation}
so the two states would commute on $S$; on $S^\perp$ both are scalar multiples of the identity, hence they would commute globally.  This contradicts
\begin{equation}
[\rho_0^\lambda,\rho_1^\lambda]
=\lambda^2[\ketbra{\psi_0},\ketbra{\psi_1}]\neq0.
\label{eq:highd-noncomm-lambda}
\end{equation}
Thus the spectral decomposition of $A_0^\lambda$ consists of two one-dimensional eigenspaces contained in $S$ and the eigenvalue-$1$ eigenspace $S^\perp$.

We next translate this spectral structure into a statement about \emph{every} finite BS-optimal trajectory, not only the canonical common-basis construction.  Let $\xi_\lambda$ be any finite trajectory representation satisfying
\begin{equation}
I(X{:}K)=B(\cE_\lambda).
\end{equation}
Write $\mu_x(k)=P(K=k|X=x)$ and $\mu(k)=P(K=k)$, discarding any zero-marginal atoms.  By \cref{prop:BS-lower},
\begin{equation}
\DKL(\mu_x\|\mu)
\ge
\DBS(\rho_x^\lambda\|\bar\rho^\lambda)
\end{equation}
for both $x=0,1$.  Since the positive-prior weighted sum of these two inequalities is an equality, each pairwise inequality is itself an equality.  Applying \cref{lem:BS-equality-support} to the pair $(\rho_0^\lambda,\bar\rho^\lambda)$ shows that every source atom lies either in $S$ or in $S^\perp$: the normalization by $(\bar\rho^\lambda)^{-1/2}$ preserves the decomposition $S\oplus S^\perp$.  Moreover, the atoms in $S$ can lie only on the two pure directions induced by the two one-dimensional eigenspaces of $A_0^\lambda|_S$.

Because
\begin{equation}
\bar\rho^\lambda|_{S^\perp}
=\frac{1-\lambda}{d}\id_{S^\perp}>0
\end{equation}
and $d\ge3$, at least one source atom has positive marginal probability and is supported in $S^\perp$.  Fix such an atom
\begin{equation}
P_u=\ketbra{u},
\qquad u\in S^\perp.
\end{equation}
Let $a=\lambda'/\lambda\in(0,1)$.  Its required one-step barycenter is
\begin{equation}
Y=\cD_a(P_u)
=aP_u+(1-a)\frac{\id}{d},
\end{equation}
whose restriction to $S$ is
\begin{equation}
Y|_S=\frac{1-a}{d}\id_S.
\label{eq:Y-S-block}
\end{equation}

Suppose, for contradiction, that a causal refinement of $\xi_\lambda$ attains the BS value at the target level.  By \cref{thm:V-properties}, an optimizer for $V_{\xi_\lambda}(a)$ exists and may be chosen with finite support.  If duplicate classical output labels carry the same pure projector, coarse-grain them to that projector; this preserves all branch barycenters and cannot increase mutual information.  The BS lower bound then forces the coarse-grained representation to remain optimal.  We may therefore work with a finite pure-state target trajectory $L$ satisfying
\begin{equation}
I(X{:}L)=B(\cE_{\lambda'}).
\end{equation}
As above, both pairwise BS inequalities must be equalities.  Applying \cref{lem:BS-equality-support} at $\lambda'$ shows that every target atom lies either in $S$ or in $S^\perp$, and that there are only two possible target pure-state directions in $S$; denote their projectors by $Q_1,Q_2$.

Condition now on the fixed source branch $P_u$.  Target atoms in $S^\perp$ contribute nothing to the $S$ block of its barycenter.  Therefore \eqref{eq:Y-S-block} forces
\begin{equation}
w_1Q_1+w_2Q_2
=\frac{1-a}{d}\id_S
\label{eq:highd-two-projectors}
\end{equation}
for some branch probabilities $w_1,w_2\ge0$.  Taking traces gives
\begin{equation}
w_1+w_2=\frac{2(1-a)}{d}>0.
\end{equation}
After normalization, \eqref{eq:highd-two-projectors} is a convex combination of two pure qubit states equal to $\id_S/2$.  The Bloch vector of $\id_S/2$ is zero, so two unit Bloch vectors can average to zero only with equal weights and opposite directions.  Hence
\begin{equation}
Q_1Q_2=0,
\qquad
w_1=w_2=\frac{1-a}{d}.
\end{equation}
Thus $Q_1,Q_2$ form an orthogonal basis of $S$.

But every target atom in $S$ is one of $Q_1,Q_2$, while every remaining target atom lies in $S^\perp$.  Consequently both $\rho_0^{\lambda'}|_S$ and $\rho_1^{\lambda'}|_S$ are diagonal in the same orthogonal basis $\{Q_1,Q_2\}$, and on $S^\perp$ both states are proportional to the identity.  They therefore commute globally.  This is impossible because
\begin{equation}
[\rho_0^{\lambda'},\rho_1^{\lambda'}]
=\lambda'^2[\ketbra{\psi_0},\ketbra{\psi_1}]\neq0.
\end{equation}
We have shown that no BS-optimal source trajectory can be causally refined into a BS-optimal target trajectory.

Finally, \cref{thm:V-properties} gives attainment of $V_{\xi_\lambda}(a)$, while \cref{prop:BS-lower} gives
\begin{equation}
V_{\xi_\lambda}(a)\ge B(\cE_{\lambda'}).
\end{equation}
If equality held, the attained optimizer would be precisely the BS-optimal target trajectory ruled out above.  Therefore the inequality is strict, proving \eqref{eq:highd-strict}.

\section{Finite moment formulation and operator dual}
\label{app:dual}
We prove \cref{thm:finite-primal} and the following dual characterization.  For Hermitian $Y_k$ and scalar $\beta$, define
\begin{equation}
 \beta+\lamax\!\left(\sum_k\alpha_kY_k\right)\le g(\alpha)\qquad\forall\alpha\in\Delta_m.
\label{eq:dual-constraint}
\end{equation}
\begin{theorem}[No-gap operator dual]
\label{thm:dual}
\begin{equation}
 V_\xi(a)=\sup_{\beta,\{Y_k\}}\left\{\beta+\sum_k\Tr[Y_kM_k(a)]:\ \eqref{eq:dual-constraint}\ \text{holds}\right\}.
\label{eq:dual}
\end{equation}
\end{theorem}
Let $\pi_k=P(K=k)$ and define the posterior vectors
\begin{equation}
\eta^k=(P(X=x|K=k))_x.
\end{equation}
Consider any finite output record $L$.  Write
\begin{equation}
\tau_\ell=P(L=\ell),
\qquad
\alpha_{\ell k}=P(K=k|L=\ell),
\end{equation}
and let $Q_\ell$ be the pure state associated with output record $\ell$.  The Markov relation $X\to K\to L$ gives
\begin{equation}
P(X=\cdot|L=\ell)=\sum_k\alpha_{\ell k}\eta^k.
\end{equation}
Hence
\begin{equation}
I(X{:}L)=\sum_\ell\tau_\ell g(\alpha_\ell),
\end{equation}
with $g$ as in \eqref{eq:g-alpha}.

The joint barycenter for source branch $k$ is
\begin{align}
\sum_\ell\tau_\ell\alpha_{\ell k}Q_\ell
&=P(K=k)\sum_\ell P(L=\ell|K=k)Q_\ell\\
&=\pi_k\cD_a(P_k),
\end{align}
which proves feasibility of \eqref{eq:finite-primal-constraints}.

Conversely, suppose finite atoms satisfy \eqref{eq:finite-primal-constraints}.  Taking traces yields
\begin{equation}
\sum_\ell\tau_\ell\alpha_{\ell k}=\pi_k.
\end{equation}
Summing over $k$ and using $\sum_k\alpha_{\ell k}=1$ gives $\sum_\ell\tau_\ell=1$, so probability normalization is already implicit in the matrix constraints.  For $\pi_k>0$, define
\begin{equation}
W(\ell|k)=\frac{\tau_\ell\alpha_{\ell k}}{\pi_k}.
\end{equation}
Then $W(\cdot|k)$ is stochastic and
\begin{equation}
\sum_\ell W(\ell|k)Q_\ell=\cD_a(P_k).
\end{equation}
Thus $W$ is a causal refinement and Bayes' rule recovers the stated $\alpha_{\ell k}$.  If several indices $\ell$ carry the same projector $Q$, merge them into one atom with total mass $\tau=\sum_\ell\tau_\ell$ and posterior $\alpha=(\sum_\ell\tau_\ell\alpha_\ell)/\tau$.  All matrix moments are unchanged, while convexity of $g$ gives $\tau g(\alpha)\le\sum_\ell\tau_\ell g(\alpha_\ell)$.  Hence redundant classical labels cannot lower the optimum, and the finite program is exactly equivalent to the projector-valued record in \cref{def:V}.  This proves the equivalence.

For an arbitrary feasible family of probability measures $\{\nu_k\}$, let
\begin{equation}
\tau(\dd Q):=\sum_k\pi_k\nu_k(\dd Q).
\end{equation}
Because $\pi_k\nu_k\ll\tau$, the Radon--Nikodym derivatives
\begin{equation}
\alpha_k(Q):=\frac{\dd(\pi_k\nu_k)}{\dd\tau}(Q)
\end{equation}
exist and satisfy $\alpha(Q)\in\Delta_m$ for $\tau$-almost every $Q$.  The Markov relation gives
\begin{equation}
P(X=\cdot|Q)=\sum_k\alpha_k(Q)\eta^k,
\end{equation}
so $I(X{:}Q)=\int g(\alpha(Q))\,\tau(\dd Q)$.  Pushing $\tau$ forward by $Q\mapsto(Q,\alpha(Q))$ therefore converts the original refinement, with the same objective and moments, into a probability measure on the compact atom space
\begin{equation}
\mathcal Z=\mathsf P_d\times\Delta_m.
\end{equation}
Conversely, let $\widetilde\tau$ be any probability measure on $\mathcal Z$ satisfying the moment constraints.  For each Borel set $A\subseteq\mathsf P_d$, define branch measures by
\begin{equation}
\pi_k\nu_k(A):=\int_{A\times\Delta_m}\alpha_k\,\widetilde\tau(\dd Q,\dd\alpha).
\end{equation}
The trace constraints make each $\nu_k$ a probability measure with the required barycenter.  Conditional on $Q$, the actual posterior is the conditional mean $\bar\alpha(Q)=\mathbb E_{\widetilde\tau}[\alpha|Q]$, so convexity of $g$ gives
\begin{equation}
I(X{:}Q)=\int g(\bar\alpha(Q))\,\widetilde\tau_Q(\dd Q)
\le\int g(\alpha)\,\widetilde\tau(\dd Q,\dd\alpha).
\end{equation}
Thus the lifted program cannot achieve a value below the original problem after projection, while every original refinement embeds in the lifted program with equal value.  The two optima are therefore identical.

Define the continuous moment map
\begin{equation}
F(z)=(\alpha_1Q,\ldots,\alpha_mQ)
\end{equation}
and scalar cost $g(\alpha)$.  The primal is a linear optimization over probability measures $\tau$ on $\mathcal Z$:
\begin{equation}
\min_\tau\int g(\alpha)\,\tau(\dd z)
\quad\text{s.t.}\quad
\int\alpha_kQ\,\tau(\dd z)=M_k(a).
\end{equation}
Compactness and continuity imply existence.  Applying Carath\'eodory to the vector $(F(z),g(\alpha))$ gives the support bound $md^2+2$ used in the main text.

For the dual, introduce Hermitian multipliers $Y_k$ and a multiplier $\beta$ for probability normalization.  The Lagrangian is
\begin{align}
\mathcal L(\tau,\beta,Y)
&=\beta+\sum_k\Tr[Y_kM_k(a)]\\
&\quad+\int\left[
 g(\alpha)-\beta-\sum_k\alpha_k\Tr(Y_kQ)
\right]\tau(\dd z).
\end{align}
The infimum over nonnegative measures is finite exactly when
\begin{equation}
\beta+\sum_k\alpha_k\Tr(Y_kQ)\le g(\alpha)
\end{equation}
for every pure $Q$ and every $\alpha\in\Delta_m$.  Maximizing the left-hand side over pure $Q$ gives
\begin{equation}
\sup_{Q\ \mathrm{pure}}\Tr\!\left[Q\sum_k\alpha_kY_k\right]
=\lamax\!\left(\sum_k\alpha_kY_k\right),
\end{equation}
which yields \eqref{eq:dual-constraint} and weak duality.

To show no gap, consider the compact convex set
\begin{equation}
\mathcal C
=\conv\{(F(z),g(\alpha)):z\in\mathcal Z\}
\end{equation}
in the finite-dimensional real vector space of $m$ Hermitian matrices plus one scalar.  Compactness follows because the generating set is compact and convex hulls of compact sets are compact in finite dimension.  Since the primal optimum is attained, $(M_1(a),\ldots,M_m(a),V_\xi(a))\in\mathcal C$, whereas for every $\varepsilon>0$ the point
$(M_1(a),\ldots,M_m(a),V_\xi(a)-\varepsilon)$ lies outside $\mathcal C$.  Strong separation therefore gives Hermitian coefficients $H_k$, a scalar $s$, and $\gamma$ such that
\begin{equation}
 \sum_k\Tr[H_kM_k(a)]+s(V_\xi(a)-\varepsilon)<\gamma
 \le \sum_k\alpha_k\Tr(H_kQ)+s\,g(\alpha)
 \quad\forall(Q,\alpha)\in\mathcal Z.
\end{equation}
Applying the right inequality to any feasible point of $\mathcal C$ with moments $M_k(a)$ and cost $V_\xi(a)$ shows $s\varepsilon>0$, hence $s>0$.  Divide by $s$, set $Y_k=-H_k/s$ and $\beta=\gamma/s$, and obtain
\begin{equation}
 \beta+\sum_k\alpha_k\Tr(Y_kQ)\le g(\alpha)
 \quad\forall(Q,\alpha)\in\mathcal Z,
\end{equation}
so $(\beta,\{Y_k\})$ is dual feasible.  The strict separation inequality also gives
\begin{equation}
 \beta+\sum_k\Tr[Y_kM_k(a)]>V_\xi(a)-\varepsilon.
\end{equation}
Thus the dual supremum is at least $V_\xi(a)-\varepsilon$ for every $\varepsilon>0$.  Letting $\varepsilon\downarrow0$ and combining with weak duality proves the no-gap identity \eqref{eq:dual}.  The theorem asserts equality of primal value and dual supremum; dual attainment is not needed.

\section{Numerical-validation details}
\label{app:numerics}
The numerical validation uses natural logarithms throughout.  The restricted-primal solver is CLARABEL with absolute, relative, and feasibility tolerances $2\times10^{-9}$ and a maximum of 2000 iterations.  In the one-step validation suite, 55 instances returned \texttt{optimal} and two returned \texttt{optimal\_inaccurate}; the maximum barycenter residual was $1.575\times10^{-7}$.  The BS-representation reconstruction residual was at most $7.468\times10^{-16}$.  A fixed Haar seed 20260813 generated the nested qutrit libraries.  Reverse-learning experiments use fixed trajectories and ten fixed seeds; model selection is based on validation cross entropy rather than on the static-versus-causal effect size.

\begin{table*}[t]
\centering
\caption{Complete predeclared qutrit grid.  $\widehat V_0$ and $\widehat V_{256}$ are restricted-primal upper approximations using the feasibility library alone and the same library augmented with 256 fixed Haar states.  The final two columns are the minimum and maximum source-atom residuals for direct convex-hull closure into the target BS basis.  Strict $V>B$ follows analytically from \cref{thm:highd-nogo}; $\widehat V_{256}-B$ is not used as a proof of the gap.}
\label{tab:qutrit-grid}
\scriptsize
\begin{tabular}{ccccrrrrr}
\toprule
$\theta$ & $\lambda$ & $r$ & $B$ & $\widehat V_0$ & $\widehat V_{256}$ & $\widehat V_{256}-B$ & closure min & closure max\\
\midrule
$\pi/3$ & 0.70 & 0.4 & 0.040778 & 0.041315 & 0.041315 & 0.000538 & 0.005757 & 0.051810\\
$\pi/3$ & 0.70 & 0.7 & 0.125276 & 0.127118 & 0.126784 & 0.001508 & 0.004572 & 0.041152\\
$\pi/3$ & 0.85 & 0.4 & 0.059898 & 0.060417 & 0.060417 & 0.000519 & 0.003219 & 0.061163\\
$\pi/3$ & 0.85 & 0.7 & 0.188172 & 0.189895 & 0.189590 & 0.001418 & 0.002511 & 0.047708\\
$\pi/3$ & 0.95 & 0.4 & 0.074791 & 0.075041 & 0.075041 & 0.000250 & 0.001137 & 0.067101\\
$\pi/3$ & 0.95 & 0.7 & 0.239562 & 0.240397 & 0.240253 & 0.000691 & 0.000877 & 0.051739\\
$\pi/4$ & 0.70 & 0.4 & 0.028013 & 0.028693 & 0.028685 & 0.000672 & 0.008096 & 0.072865\\
$\pi/4$ & 0.70 & 0.7 & 0.090956 & 0.093009 & 0.092822 & 0.001866 & 0.006377 & 0.057393\\
$\pi/4$ & 0.85 & 0.4 & 0.041689 & 0.042338 & 0.042314 & 0.000625 & 0.004517 & 0.085831\\
$\pi/4$ & 0.85 & 0.7 & 0.141985 & 0.144034 & 0.143885 & 0.001900 & 0.003485 & 0.066224\\
$\pi/4$ & 0.95 & 0.4 & 0.052573 & 0.052884 & 0.052871 & 0.000299 & 0.001594 & 0.094018\\
$\pi/4$ & 0.95 & 0.7 & 0.186348 & 0.187398 & 0.187332 & 0.000984 & 0.001213 & 0.071591\\
\bottomrule
\end{tabular}
\end{table*}

\paragraph{High-dimensional causal-correction bridge.}
For the same 12 predeclared qutrit settings, we additionally constructed explicit source-blind one-step refinements using nested fixed Haar libraries with $0,64,128,256,512$, and, when the solver succeeded, $1024$ random atoms.  The primary setting $(\pi/4,0.85,0.7)$ decreased from $\widehat V_0=0.144034$ to $\widehat V_{512}=0.143885$, with the final $256\to512$ change only $3.06\times10^{-6}$ nats compared with the $1.90\times10^{-3}$ feasible causal correction.  The $1024$-atom primary solve failed under CLARABEL and is not used.  We do not report a rigorous numerical dual certificate; all finite-library values are therefore treated strictly as feasible upper approximations, while analytic strictness follows from \cref{thm:highd-nogo}.

For each completed refinement, the optimized column-stochastic transition kernel was revalidated branch by branch.  In the primary setting the maximum barycenter residual is $4.07\times10^{-9}$, while the exact joint law built from the feasible transition gives
\[
 I(X{:}K_{\rm src})-I(X{:}K_{\rm tgt})
 =I(X{:}K_{\rm src}\mid K_{\rm tgt})
 =\mathrm{CE}_{\rm noX}^\star-\mathrm{CE}_{\rm withX}^\star
 =0.223874,
\]
up to machine precision.  By contrast, independently substituting the static BS optimum would predict $0.225775$.  Thus the finite causal construction operationalizes the high-dimensional obstruction without treating its finite-library gap as an exact numerical value of $V-B$.

\paragraph{High-dimensional multi-step scheduling.}
For the primary noncommuting qutrit family we initialize at cumulative retention $\lambda_0=0.85$ and construct a four-step path to the fully mixed endpoint.  The density-static schedule selects $\lambda_t$ by linearly spacing the BS values $B(\cE_t)$, while still propagating every step from the actual finite trajectory produced at the preceding level.  The causal schedule instead searches the current finite-library refinement problem at each step for a retention whose realized information follows $J_t=(1-t/4)J_0$.  We repeat the full recursive construction independently with fixed nested Haar prefixes of size $128$, $256$, and $512$; no retention path is inherited from a smaller library.  Every selected transition is revalidated branch by branch, with maximum barycenter residual $1.37\times10^{-8}$ over the independently optimized causal paths.  The resulting static-versus-causal summary is reported in \cref{tab:qutrit-multistep}.  Although the finite-library retention coordinates adapt as the library grows, all three causal trajectories realize nearly constant information decrements and substantially reduce the worst step relative to density-static equalization.

\paragraph{High-dimensional reverse learning.}
We persist the complete $N=256$ and $N=512$ static and causal qutrit trajectories, including their state-resolved supports and transition kernels, and validate the Markov propagation, density reconstruction, branchwise barycenters, and exact Bayes reverse identity before learning.  The trajectories, schedules, Haar seed, and exact burdens are then frozen.  Predictor capacity, learning rate, sample budget, and the smoothing strength of the empirical estimator are selected by validation weighted cross entropy only; all ten fixed seeds are retained.  The final neural model uses timestep-specific categorical reverse heads, which remove interference between the different finite state alphabets at successive steps.

The final predictor reaches burden MAE $9.72\times10^{-4}$ nats at $N=256$ and $1.36\times10^{-3}$ nats at $N=512$, with mean excess cross entropy above the Bayes references $1.08\times10^{-3}$ and $1.75\times10^{-3}$ nats, respectively.  The learned burden therefore closely tracks the realized information geometry: CV reductions are $95.19\%$ and $92.02\%$, and worst-step reductions are $25.37\%$ and $16.29\%$.  Both comparisons have the same direction in all ten seeds.  The model comparison in \cref{tab:qutrit-learning-models} further shows that a learned categorical table and a validation-selected empirical conditional-frequency estimator reproduce the same schedule ordering.

\begin{table*}[t]
\centering
\caption{Qutrit reverse-learning model comparison on the frozen $N=256/512$ trajectories.  MAE is between learned and exact stepwise burdens; excess CE is the mean held-out cross entropy above the Bayes reference.  CV and worst reductions compare the causal schedule against density-static BS equalization.  All entries retain the same direction in $10/10$ seeds.}
\label{tab:qutrit-learning-models}
\scriptsize
\begin{tabular}{llrrrr}
\toprule
Model & $N$ & burden MAE & excess CE & learned CV red. & learned worst red. \\
\midrule
Empirical frequency & 256 & 0.001279 & 0.000898 & 93.87\% & 25.10\% \\
Empirical frequency & 512 & 0.001860 & 0.001714 & 90.08\% & 15.56\% \\
Improved MLP & 256 & \textbf{0.000972} & 0.001082 & \textbf{95.19\%} & \textbf{25.37\%} \\
Improved MLP & 512 & \textbf{0.001362} & 0.001754 & \textbf{92.02\%} & \textbf{16.29\%} \\
Learned categorical table & 256 & 0.001185 & 0.006808 & 94.35\% & 25.04\% \\
Learned categorical table & 512 & 0.001702 & 0.008406 & 90.58\% & 15.86\% \\
\bottomrule
\end{tabular}
\end{table*}

\paragraph{End-to-end rollout and propagation diagnostics.}
We evaluate fully source-blind reverse rollouts on the frozen $N=256/512$ trajectories using the same $4096$-trajectory budget, architecture, optimizer, and ten paired seeds. All trajectory and reverse-kernel checks pass; the largest information-identity error is $6.8\times10^{-16}$. Exact probability-vector propagation removes rollout sampling noise.

The diagnostics isolate two finite-model effects beyond the exact burden $c_t$. At $N=256$, channel-constrained scheduling has larger weighted TV/KL mismatch at every step and ends at trace distance $0.00921$ versus $0.00753$; its contraction advantage does not offset the local error. At $N=512$, local mismatch is again larger, but the endpoint improves to $0.00804$ versus $0.01253$. Single-step hybrids identify stronger contraction as the mechanism: the isolated $t=4$ perturbation produces endpoint error $0.00379$ versus $0.00691$ (paired difference $-0.00312$, 95\% CI $[-0.00591,-0.00076]$), with TV amplification $0.0517$ versus $0.1133$; $t=2$ shows the same advantage. The exact burden, finite-model mismatch, and trajectory-dependent propagation are therefore distinct components of learned rollout behavior.

\begin{table}[t]
\centering
\caption{Source-blind qutrit rollout diagnostics under exact probability propagation. Endpoint trace distance uses all four learned reverse steps; isolated $t=4$ uses the learned $t=4$ kernel and exact kernels elsewhere.}
\label{tab:qutrit-rollout-diagnostic}
\scriptsize
\begin{tabular}{crrrr}
\toprule
$N$ & endpoint CC & endpoint static & isolated $t=4$ CC & isolated $t=4$ static \\
\midrule
256 & 0.00921 & 0.00753 & 0.00506 & 0.00500 \\
512 & \textbf{0.00804} & 0.01253 & \textbf{0.00379} & 0.00691 \\
\bottomrule
\end{tabular}
\end{table}

For the qubit clock validation, the source retention is $0.9$ and nine target levels span $0$ through $0.8$ in increments of $0.1$.  The maximum absolute discrepancy between the restricted primal and the exact BS value is $3.533\times10^{-9}$ nats.

For the theorem-aligned eight-step scheduling and reverse-learning validation, the same equal-prior binary-qubit source pair is initialized at cumulative retention $\lambda_0=0.9$ and terminated at $\lambda_8=0$.  This gives $J_0=0.456368176$ and $J_0/T=0.057046022$ nats.  The exact BS-optimal branches and binary-symmetric transition kernels are propagated analytically.  Maximum propagation, branchwise barycenter, and state-reconstruction residuals are $3.331\times10^{-16}$, $1.889\times10^{-15}$, and $5.769\times10^{-16}$, respectively; the two independent evaluations of $c_t$ and the Bayes cross-entropy gap agree to $3.053\times10^{-16}$.

The learned validation uses ten fixed seeds.  Each schedule provides $4096$ sampled training trajectories and $10^5$ held-out trajectories.  Both reverse predictors use categorical one-hot inputs, one $\tanh$ hidden layer of width $32$, and one Bernoulli logit.  The source-blind model receives $(K_t,t)$ and the source-aware reference additionally receives $X$.  Both are optimized with Adam at learning rate $0.01$ for $1200$ updates with batch size $512$, with paired initializations across schedules.  The resulting control experiment is summarized in \cref{tab:qubit-learned-control}.

\begin{table}[t]
\centering
\caption{Analytic binary-qubit control. Learned quantities are means over ten fixed seeds.  MAE is $\frac1T\sum_t|\widehat c_t-c_t|$.}
\label{tab:qubit-learned-control}
\scriptsize
\begin{tabular}{lrrrrr}
\toprule
Schedule & $\max c_t/(J_0/T)$ & CV$(c_t)$ & $\max\widehat c_t$ & CV$(\widehat c_t)$ & MAE \\
\midrule
Equal information & 1.000 & 0.000 & 0.05938 & 0.0287 & 0.00151 \\
Linear retention & 2.590 & 0.792 & 0.14757 & 0.8105 & 0.00217 \\
Cosine retention & 2.068 & 0.762 & 0.11791 & 0.7688 & 0.00189 \\
\bottomrule
\end{tabular}
\end{table}

\section{Additional consequences and boundary cases}
\label{app:additional}
\subsection{Scheduling bracket}
Let $\rho_x(a)=\cD_a(\rho_x)$ and $B_\xi(a)=B(\{p_x,\rho_x(a)\})$. Whenever the average state is faithful,
\begin{equation}
B_\xi(a)\le V_\xi(a)\le aJ.
\label{eq:V-sandwich}
\end{equation}
For $0<\tau<J$, define $a_{\rm BS}^\star(\tau)=\sup\{a\in[0,1):B_\xi(a)\le\tau\}$. Then
\begin{equation}
\frac{\tau}{J}\le a_\xi^\star(\tau)\le a_{\rm BS}^\star(\tau).
\label{eq:retention-sandwich}
\end{equation}
The left endpoint follows from the keep-or-reset upper bound and the right from $B_\xi\le V_\xi$; the right inequality is an equality along the BS-optimal binary-qubit trajectory.

\subsection{Clock-defect identity}
For any realized trajectory with $B_t:=B(\cE_t)$, set $\Gamma_t:=J_t-B_t\ge0$. Then
\begin{equation}
c_t=(B_{t-1}-B_t)+(\Gamma_{t-1}-\Gamma_t).
\label{eq:clock-defect}
\end{equation}
\subsection{Commuting ensembles in arbitrary dimension}
If all density states commute, choose a common orthonormal basis $\{P_j\}_{j=1}^{d}$.  Depolarization acts as
\begin{equation}
\cD_a(P_j)
=aP_j+\frac{1-a}{d}\sum_{i=1}^{d}P_i
=\sum_i\left(a\delta_{ij}+\frac{1-a}{d}\right)P_i.
\label{eq:commuting-kernel}
\end{equation}
Thus the basis simplex is closed under the forward channel with transition matrix \eqref{eq:commuting-kernel}.  If $q_x(j)$ are the diagonal coefficients and $\bar q(j)=\sum_xp_xq_x(j)$, the basis record has $I(X{:}J)=\sum_xp_x\DKL(q_x\|\bar q)=B(\cE)$, because the BS divergence reduces to classical KL on commuting states.  Hence this representation is BS-optimal and remains dynamically closed; in this classical subfamily, no causal correction is needed.

\subsection{Clock-defect bounds}
Equation~\eqref{eq:clock-defect} immediately quantifies the error made by scheduling with $B$ instead of the realized trajectory information.  If a chosen BS grid has constant decrement $B_{t-1}-B_t=c$ and $\Gamma_0=\Gamma_T=0$, then
\begin{equation}
c_t-c=\Gamma_{t-1}-\Gamma_t.
\end{equation}
If $M=\max_t\Gamma_t$, the total variation of a path that starts and ends at zero gives
\begin{equation}
\sum_{t=1}^{T}|c_t-c|\ge2M.
\end{equation}
By Cauchy--Schwarz,
\begin{equation}
\sum_{t=1}^{T}(c_t-c)^2\ge\frac{4M^2}{T},
\qquad
\max_t|c_t-c|\ge\frac{2M}{T}.
\end{equation}
These bounds are not needed for the main scheduling theorem because $V$ directly corrects the clock, but they make explicit how a density-only schedule can misallocate reverse difficulty when the realization gap varies across time.

\subsection{Faithful references and rank-deficient endpoints}
The definition \eqref{eq:DBS} only requires the reference state to be faithful; the first argument may be rank deficient.  For any ensemble average $\bar\rho$ and retention $a<1$,
\begin{equation}
\cD_a(\bar\rho)\succeq\frac{1-a}{d}\id>0,
\end{equation}
so every interior BS comparison used in the paper is covered directly, even when individual data states are rank deficient.  The channel-constrained value $V$ itself has no faithfulness requirement.  If the ensemble average is rank deficient at an endpoint, $B$ may be evaluated after restricting to its support (when the support condition is satisfied) or by taking the faithful interior limit; none of the causal scheduling results depends on selecting an endpoint convention.

\subsection{Beyond depolarization}
For an arbitrary CPTP map $\Phi$, one can define
\begin{equation}
V_\xi(\Phi)
:=\inf_{\Lambda(\nu_k)=\Phi(P_k)}I(X{:}K').
\end{equation}
The reverse-information identity and finite moment/dual formulations continue to hold.  The full schedule inversion in the main text additionally uses an ordered one-parameter family with: (i) a composition order that makes stronger noise a post-processing of weaker noise, (ii) a fully mixing endpoint with value zero, and (iii) continuity in the channel parameter.  Depolarization satisfies all three properties exactly.  Extending the least-noise scheduling theorem to other quantum noise families therefore reduces to verifying these channel-order properties rather than redefining the information objective.

\end{document}